\documentclass[a4paper,fleqn]{cas-sc}

\usepackage{amsmath,amsthm,amsfonts,amssymb,bbm,xcolor,mathtools}
\usepackage{etoolbox}
\usepackage{xcolor}
\usepackage{xspace}
\usepackage{etoolbox}
\usepackage{svg}
\usepackage{microtype}

\usepackage{tikz}
\usetikzlibrary{arrows.meta,shapes}
\usetikzlibrary{shapes.geometric, circuits.logic.US, positioning}
\usetikzlibrary{decorations,decorations.pathreplacing}
\usepackage{subfig}
\usepackage{complexity}

\usepackage[authoryear,longnamesfirst]{natbib}

\def\tsc#1{\csdef{#1}{\textsc{\lowercase{#1}}\xspace}}
\tsc{WGM}
\tsc{QE}
\usepackage{hyperref}
\usepackage{cleveref}

\theoremstyle{plain}
\newtheorem{theorem}{Theorem}[section]

\newtheorem{definition} [theorem]{Definition}

\newtheorem{lemma}      [theorem]{Lemma}
\newtheorem{problem}    [theorem]{Problem}
\newtheorem{proposition}[theorem]{Proposition}

\newtheorem{claim}      [theorem]{Claim}

\newenvironment{claimproof}[1]{\par\noindent{\textit{Proof of claim.}}\space#1}{\hfill $\diamond$\vspace{0.3cm}}

\newcommand{\ie}{\emph{i.e.,}\@\xspace}

\newcommand{\majloc}{\texttt{maj}}
\newcommand{\majglob}[1]{A}
\newcommand{\majority}  [1] {\majloc(#1) }

\newcommand{\iteration} [3] {#1^{#2}(#3) }
\newcommand{\nbx}       [2] {#1(#2, x) }
\newcommand{\nbAx}      [2] {#1(#2, A(x)) }
\newcommand{\nbAit}     [3] {#1(#2, A^{#3}(x)) }
\newcommand{\val}       [2] {#1_{#2} }

\newcommand{\pinf}{\lfloor n / 2 \rfloor}
\newcommand{\psup}{\lceil n / 2 \rceil}

\definecolor{myGreen}{RGB}{0,120,0}
\definecolor{myRed}{RGB}{200,0,0}

\begin{document}

\shorttitle{Example of MBAN solving DCT}    

\shortauthors{}  

\title [mode = title]{Non-trivial automata networks do exist that solve the global majority problem with the local majority rule}

\author[1]{Pedro Paulo Balbi}
\ead{pedrob@mackenzie.br}

\author[2]{K\'{e}vin Perrot}
\ead{kevin.perrot@lis-lab.fr}

\author[2]{Marius Rolland}
\ead{marius.rolland@lis-lab.fr}
\cormark[1]
\cormark[1]

\author[1]{Eurico Ruivo}
\ead{eurico.ruivo@mackenzie.br}

\affiliation[1]{organization={Universidade Presbiteriana Mackenzie, Faculdade de Computa\c{c}\~{a}o e Inform\'{a}tica},
		addressline={Rua da Consola\c{c}\~{a}o 896}, 
		city={Consola\c{c}\~{a}o, S\~{a}o Paulo},
		postcode={SP 01302-907}, 
		country={Brazil}}

\affiliation[2]{organization={Aix-Marseille Universit{\'e}, CNRS, LIS, UMR 7020}, addressline={Av. de Luminy 163}, city={Marseille}, postcode={13009}, country={France}}

\begin{abstract}
The global majority problem, often referred to as the Density Classification Task, is a classical benchmark in the context of probing the computational capabilities of automata networks. It poses the simple yet challenging problem of determining, by totally local means, whether an arbitrary initial configuration of binary states can evolve to a final, homogeneous global configuration that reflects the initial global majority. Although it is known that in the specific case of cellular automata with periodic boundaries no rule is able to solve the problem, in other formulations solutions are known and, in others, the problem is still open. Aligned with the latter, here we explore the possibility of solving the problem with automata networks, operating only with the local majority rule, with a focus on identifying non-trivial cases where it can be solved and explaining why they do so.
\end{abstract}


\maketitle

\section{Introduction: Density Classification Task, a challenge in computing by distributed consensus}

The \textit{global majority problem}, often referred to as the \textit{Density Classification Task}, DCT for short, stands as a seminal benchmark in the study traditionally in the context of cellular automata (CAs), but naturally relevant and applicable to automata networks in general, and more broadly, complex systems and distributed computation. It poses a seemingly simple yet challenging problem: given an initial configuration of binary states (0s and 1s), can a system of locally interacting components evolve to a final, homogeneous global configuration that reflects the initial global majority? The enduring difficulty and theoretical implications of the DCT have made it a key component for understanding the limits and capabilities of decentralised information processing.

At its core, the DCT demands that an automata network determines, and then propagate, the global majority state. In its original formulation, a binary, one-dimensional CA on an odd-sized lattice with periodic boundary conditions has to converge to a final configuration of all cells being in the 1-state (a homogeneous fixed-point) if the initial configuration has more 1s than 0s, and to a configuration of all 0s, otherwise. The challenge is amplified by the fundamental operating principle that each lattice position (or the nodes of an arbitrary automata network) updates its state synchronously based only on its own current state and the states of its immediate neighbours, according to a fixed local rule. There is no central controller, no global memory, and no direct communication between distant cells, which implies a distributed consensus is being achieved within the system.

As a decision problem solved by distributed consensus in automata networks, the study of DCT has continuously unveiled a deep and rich conceptual web related to discrete dynamical systems, to a point that it has become the most fruitful benchmark problem in the area (although others also exist, as discussed in \cite{portifolio2018}). 

It is already known for many years that no single binary CA rule can solve the task (\cite{Land1995}, or \cite{negativeAsynchDCT2024} for a much simpler proof) and evidences exist that a solution may not exist even employing deterministic block sequential asynchronous updates (\cite{negativeAsynchDCT2024}). These no-go results firmly established the problem's theoretical intractability under standard CA assumptions. 

Nevertheless, the related literature demonstrates that, depending on how DCT is formulated -- such as the type of boundary condition; introduction of additional states; possibility of rules changing over time; different rules at distinct nodes of the network; other connection patterns among the nodes; addition of memory to the state transitions; and asynchronous kinds of updates -- a solution can be shown to exist or not, or at least contribute evidences towards each possibility. Many of these cases are discussed in \cite{deOliveira2014}, still the most comprehensive survey on the problem currently available. 

In tune with the latter efforts, here we explore the ability of the \textit{local majority rule} to solve DCT with distinct connection patterns among the nodes of an automata network. Such a rule simply establishes that the next state of a node is given by the current most frequent state in its neighbourhood, namely, the nodes from which it receives a connection on a directed connected pattern (including the possibility of a self-loop); in case of a draw, the node's next state remains unchanged. More precisely, what we show here is that, beyond the possibility of solving DCT with trivial cases of network connection patterns (such as the fully connected pattern, by which every node receives a connection from every other node), solutions based upon more elaborate connection patterns do exist, where only partial visibility of the nodes by the others is allowed; our focus is on four of these non-trivial cases, although others might exist. 

The reason we have chosen the local majority rule came in part from its nice association to DCT, as it stands for the local version of the global problem we are tackling. Also significant is the fact that the rule is quite pervasive in the literature related to distributed consensus, not only in the context of automata networks (as in \cite{goles1990_neural_automata, PSpaceMajRules_2016, gaertner2017_biased_majority, Abilhoa_Balbi_2020, goles2023_majority_networks}), but also in physics, social sciences and biology (\cite{krapivsky2003_dynamics, galam2008_sociophysics, marshall2019_quorums}). Finally, the known fact that the local majority rule performs badly on DCT in cellular automata would allow us to really probe the flexibility that general automata networks might bring, which was effectively what came about, given the solution families we found.

In the next section we give all the basic definitions to be used in the rest of the paper, as well as the related notation we rely on. The subsequent section is the core of the paper, where the classes of automata networks we found, for which a DCT solution does exist, are presented together with the formal argument about why they work; the presentation is given in five subsections, the first one discussing the trivial cases of network connection patterns, and the other four the non-trivial cases. In the last section of concluding remarks, we summarise our findings and consider alleys for possible follow-up research.
	
	\section{Definitions}

Let $[n]=\{0,\ldots,n-1\}$, and $V = [n]$ be a set of $n$ nodes of a graph, which for present purposes is regarded as a network of $n$ \emph{automata}. 
A \emph{configuration} assigns a binary state among $\{0,1\}$ to each automaton,
\emph{i.e.}, it is an element $x\in\{0,1\}^{|V|}$ or equivalently $x\in\{0,1\}^n$
such that $\val{x}{v}$ is the state of automaton $v\in V$.
We extend this notation to any subset $W\subseteq V$,
with $\val{x}{W}\in\{0,1\}^{|W|}$ the restriction of $x$ on domain $W$.

\begin{definition}
  A Boolean automata network (BAN) of size $n$ is a tuple of $n$ \emph{local functions}
  $f_v:\{0,1\}^n\to\{0,1\}$, one for each automaton $v\in V$,
  inducing a deterministic dynamics $f:\{0,1\}^n\to\{0,1\}^n$ on state space $\{0,1\}^n$,
  simply defined as
  $f(x)=(f_0(x),\dots,f_{n-1}(x))$ for all configurations $x\in\{0,1\}^n$.
\end{definition}

We assume that the automata are updated fully synchronously (or, in parallel) at each step.
A central object in automata network theory is its \emph{interaction graph},
encoding the effective dependencies among the automata.
It is the digraph $G_f=(V,E)$ on vertex set $V=[n]$
and with arc $(u,v)\in E$, if and only if $u$ is an essential variable of $f_v$, i.e., there exists $x\in\{0,1\}$ such that $f_v(x)\neq f_v(x+e_u)$, taken bitwise modulo two, with $e_u$ being the $u$-th base vector.

In the present work, we focus on Boolean automata networks with local functions computing
the majority on a subset of their automata
(on the in-neighbours of the vertex in the interaction graph),
and keeping the current state in case of tie.

\begin{definition}
  The local majority function on $W\subseteq V$ at $v\in V$ is 
  defined as $\majloc_W^v:\{0,1\}^n \to \{0,1\}$ such that:
  \[
    \majloc_W^v(x) = \begin{cases}
      0 &\textnormal{if } \sum_{u\in W} x_u < \frac{|W|}{2}\\
      1 &\textnormal{if } \sum_{u\in W} x_u > \frac{|W|}{2}\\
      x_v &\textnormal{otherwise (tie case).}
    \end{cases}
  \]
\end{definition}

A majority Boolean automata networks will simply be given by its interaction graph,
encoding the architecture of the dependencies (domains of the majority local functions).
For a graph $G=(V,E)$, let $N_G(v)$ denotes the set of in-neighbours of $v\in V$ in $G$.

\begin{definition}
  The majority Boolean automata network (MBAN) $A_G$ on the graph $G=(V,E)$ with $V=[n]$
  is the BAN of size $n$ defined by the local functions
  $f_v = \majloc_{N_G(v)}^v$ for each automaton $v\in V$.
\end{definition}

In the following, our constructions will try to avoid the possibility of tie cases, that is, Boolean networks where the in-degree of all automata is odd.
When there is no ambiguity, we may refer to $A_G$ simply as $A$.

The dynamics of an automata network $A$ induces a functional directed graph (where each node has out-degree exactly one) whose set of nodes is the set of configurations, and there is an arc $(x,y)$ if and only if $A(x) = y$.
A functional directed graph is a set of connected components that are composed of a cycle of periodic configurations, called \emph{limit cycle}, and additional upward trees rooted in the cycle (transient configurations). 
We say that a configuration has \emph{period} $m>0$ when $A^m(u) = u$ and no $t<m$ verifies this ($m$ is also the \emph{length} of the limit cycle).
A configuration of period $1$ is a \emph{fixed point}.

We are interested in the characterization of MBANs whose automata collectively
interact through local majority functions, and whose dynamics globally converges towards the fixed point made up by 
the majority state, from any starting configuration.
This problem is well-known in the literature, and usually cited as the density classification task (DCT), in a reference to the relation between the amounts of 1s and 0s in the initial configuration.

\begin{definition}
  A BAN $f:\{0,1\}^n\to\{0,1\}^n$ solves the \emph{density classication task} (DCT), when:
  \begin{itemize}[nosep]
    \item any configuration of $\{0,1\}^n$ with a majority of $0$s converges to the fixed point $0^n$, and
    \item any configuration of $\{0,1\}^n$ with a majority of $1$s converges to the fixed point $1^n$.
  \end{itemize}
Since the problem is well defined only for odd-sized networks, i. e., when $n$ is odd, so we restrict our considerations to this case.
\end{definition}

Observe that, since the dynamics is deterministic and the property is asked to apply to any configuration,
in order to solve the DCT the majority state should never change throughout the evolution.

For odd $n$, we simply refer to $\majloc$ for $\majloc_V:\{0,1\}^n\to\{0,1\}$.
For any subset configuration $x\in\{0,1\}^n$ and $S \subseteq V$, we denote $\nbx{T}{S} = |\{v \in S \mid \val{x}{v} = 1\}|$ and $\nbx{F}{S} = |\{v \in S \mid \val{x}{v} = 0\}|$,
and use the shortcuts $\nbx{T}{u} = \nbx{T}{N_G(u)}$ and $\nbx{F}{u} = \nbx{F}{N_G(u)}$
for the number of in-neighbors of $u$ in state $1$ and $0$, respectively.



As an example, the complete directed graph $K_n$ with $n^2$ arcs solves the majority classification task:
it converges to the correct fixed point after one step.
On the contrary, a directed cycle $C_n$ with $n$ arcs has the correct two fixed points,
but it fails to solve the majority classification task because no other configuration converges
(it cycles around the graph, because the majority local function of arity one is simply the identity).

	\section{MBANs able to solve DCT}

	\subsection{Introductory cases}

	Due to the previous section, it is hard to decide whether an MBAN can solve DCT. 
	For this reason, we propose some classes of MBANs able to solve DCT, and checkable in polynomial time (over the number of automata). 
    
	The first obvious class is the MBAN defined over the complete directed graph, that is, the graph $G = (V,E)$ such that $(u,v) \in E$ for all $u,v \in V$. 
	Indeed, since the update of all nodes is the majority of all nodes in the graph, this clearly entails that, for all configurations $x$, $A(x)$ becomes the majority of $x$. 
	
	The second class is the MBAN which are built from a digraph $G = (V,E)$ with $n$ nodes; by adding $n+1$ new nodes such that each new node has all the $2n+1$ nodes for in-neighborhood and out-neighborhood. 
    See Figure~\ref{fig:generator} for an example.
	We refer to this class as the \emph{generated MBAN}.
	Note that we can construct the graph $G$ of a generated MBAN from any other graph $G'$.
	In other words, there exists an injective function $f : G_n \to G_{2n + 1}$ such that $f(x)$ describes a generated MBAN, where $G_n$ is the set of directed graphs with $n$ nodes.
   Indeed, if we consider $G = (V, E)$ a graph, we define $S = \{s_0, \ldots, s_{|V|}\}$ a set of $|V| + 1$ new nodes, create $G' = (V', E')$ such that $V' = V \cup S$, and $E'$ is built as follows:
	\begin{enumerate}[nosep]
		\item $(u,v) \in E$ imply that $(u,v) \in E'$,
		\item $(u,s)$ for all $u \in V$ and $s \in S$,
		\item $(s,v)$ for all $s \in S$ and $v \in V$.  
	\end{enumerate}
	
	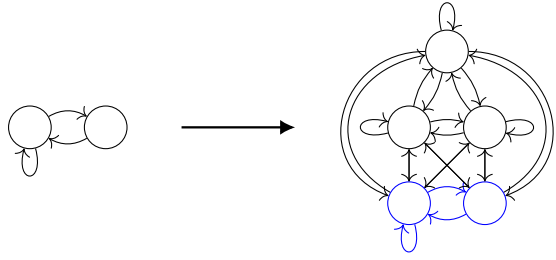
\begin{figure}
		\centering
		\begin{tikzpicture}
	\tikzstyle{automate} = [draw,circle,inner sep=0pt,minimum size=16pt]
    \tikzstyle{automateb} = [draw,circle,inner sep=0pt,minimum size=16pt, color=blue]
	\tikzstyle{arc} = [-{>[length=1mm]}]
    \tikzstyle{arcb} = [-{>[length=1mm]}, color=blue]

	\node[automate] (a0) at (-2,0.5) {};
	\node[automate] (a1) at (-1,0.5) {};
	
	\draw[line width=0.3mm, -{Latex[length=2mm, width=2mm]}] (0,0.5) -- (1.5,0.5);
	
	\node[automateb] (b0) at (3,-0.5) {};
	\node[automateb] (b1) at (4,-0.5) {};
	\node[automate] (b2) at (3,0.5) {};
	\node[automate] (b3) at (4,0.5) {};
	\node[automate] (b4) at (3.5,1.5) {};

    \draw[arc] (a0) to[bend left] (a1);
    \draw[arc] (a1) to[bend left] (a0);
    \draw[arc] (a0) to[loop below] (a0);

    \draw[arcb] (b0) to[bend left] (b1);
    \draw[arcb] (b1) to[bend left] (b0);
    \draw[arcb] (b0) to[loop below] (b0);
    
	\draw[arc] (b2) to[bend left=.5cm] (b3);
	\draw[arc] (b3) to[bend left=.5cm] (b4);
	\draw[arc] (b4) to[bend left=.5cm] (b2);
	
	\draw[arc] (b3) to[bend left=.5cm] (b2);
	\draw[arc] (b4) to[bend left=.5cm] (b3);
	\draw[arc] (b2) to[bend left=.5cm] (b4);
	
	\draw[arc] (b2) to[loop left] (b2);
	\draw[arc] (b3) to[loop right] (b3);
	\draw[arc] (b4) to[loop above] (b4);
	
	\draw[arc] (b4) to[out=180,in=160,looseness=1.5] (b0);
	
	\draw[arc] (b4) to[out=0,in=20,looseness=1.5] (b1);
	
	\draw[arc] (b1) to[out=30,in=-10,looseness=1.5] (b4);
	\draw[arc] (b0) to[out=150,in=190,looseness=1.5] (b4);
	
	\draw[arc] (b0) to (b2);
    \draw[arc] (b2) to (b0);
	\draw[arc] (b0) to (b3);
    \draw[arc] (b3) to (b0);
	\draw[arc] (b1) to (b2);
    \draw[arc] (b2) to (b1);
	\draw[arc] (b1) to (b3);
    \draw[arc] (b3) to (b1);

\end{tikzpicture}
		\caption{Example of transformation of a graph (left) into a graph of an MBAN able to solve DCT (right). The nodes and arcs in blue emphasize the original graph.}\label{fig:generator}
	\end{figure}
	\ \\
    Any generated MBAN $A$ can solve DCT.
    Intuitively, at the first iteration, the nodes of $S$ have for value the majority of the original configuration. 
    And since all nodes of $V$ have the majority of their in-neighbors in $S$, then at the end of the second iteration, each node of $V$ takes the value of the nodes of $S$. 
    More formally, the in-degree of all nodes in $S$ is $|V|$, it follows that $\val{A(x)}{s} = \majloc(x)$ for all $s \in S$. 
	Also, since $|N_{G'}(v) \cap S| =  |V| + 1 > |N_{G}(v)|$ for all $v \in V$, we have that $\val{\iteration{A}{i}{x}}{v} = \majloc(x)$ for any integer $i \ge 2$. 
	In summary:
	
	\begin{proposition}
		If $G$ is a graph of $n$ nodes, then there exists an MBAN over graph $G'$ with $2n + 1$ nodes, that admits $G$ as an induced subgraph, and is able to solve DCT.
	\end{proposition}
	A direct consequence is that the number of MBANs able to solve DCT is exponential over the number of automata.  
	
	Our goal is to find some ``simple" class of MBAN able to solve DCT. 
	For this, we start by introducing 4 metrics that would allow us to measure the difficulty of an MBAN to solve the problem.  

\begin{enumerate}
    \item The number of edges in the graph: we consider that the higher the number of edges in the network (dependency between the automata), the more difficult for the network. 
	\item The number of different in-degrees: this is meant to approximate the number of different kinds of neighbourhoods in the MBAN. 
	\item The maximum in-degree of the graph: this measures the level of locality of an MBAN.
	\item The number of iterations needed for convergence. 
\end{enumerate}
    
	With these metrics in mind, we introduce the class of MBANs whose graph $G = (V,E)$ is a path rooted in $v_0$ and has an edge from each $v \in V$ to $v_0$. 
    For some intuitive reason, this kind of network can solve DCT. 
    Indeed, at the end of the first iteration, $v_0$ contains the majority of the original configuration (because all nodes of $V$ are in its in-neighborhood), and we propagate this value along the path. 
	Such a graph can also be regarded as a cycle of length $n$ that contains all the cycles of length between $1$ and $n-1$, that is, letting $C_i$ be the cycle of length $i$ in $G$, then this graph is such that $\bigcap_{i=1}^{n} C_i = C_1$. This is why we refer to this class as the one with \emph{complete cycle MBAN}; Figure~\ref{fig:complete_cycle} provides an example. 
 
    The following algorithm generates this kind of network from any size $n$: 

    \begin{enumerate}[nosep]
        \item Create $n$ nodes labeled from $0$ to $n-1$;
        \item For each node $i$ add an arc from $i$ to $i+1 \mod n$;
        \item For each node $i$ add an arc from $i$ to $0$.
    \end{enumerate}

	\begin{figure}
    	\centering
		\begin{tikzpicture}
	\tikzstyle{automate} = [draw,circle,inner sep=0pt,minimum size=16pt]
	\tikzstyle{arc} = [-{>[length=1mm]}]
	
	\def\sommets{{4,3,2,1,0,6,5}} 
	
	\foreach \i in {0,...,6} {
		\pgfmathtruncatemacro{\nexti}{\sommets[\i]};
		\node[automate] (A\nexti) at ({51.4286*\i}:2) {\pgfmathparse{\sommets[\i]}\pgfmathresult};
	}

	\foreach \i in {0,...,6} {
		\pgfmathtruncatemacro{\next}{mod(\i+1,7)};
		\draw[arc] (A\i) to[bend left,  looseness=0.5] (A\next);
	}
	
	\foreach \i in {1,...,5}{
		\draw[arc] (A\i) to[bend left, looseness=0.5] (A0);
	}
	
	\draw[arc] (A0) to[loop left] (A0);
\end{tikzpicture}
		\caption{The graph of the complete cycle MBAN of 7 nodes. For all $i$ between $0$ and $5$, the cycle $(0, \cdots, i, 0)$ is the cycle of length $i+1$.}
        \label{fig:complete_cycle}
	\end{figure}
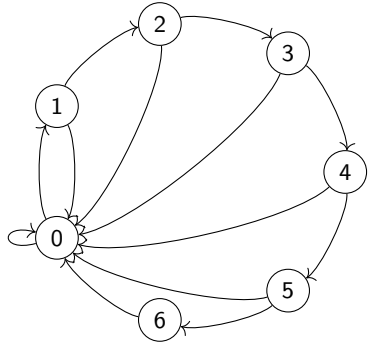
	
	\begin{lemma}\label{lemma:path_with_all}
		Let $A$ be a complete cycle MBAN.
		Then, $A$ solves DCT. 
	\end{lemma}
	
	\begin{proof}
		Let $x$ be a configuration such that $\majority{x} = 1$.
		Let $G = (V,E)$ be the graph of $A$ such that $V = \{v_1, \ldots, v_n\}$ and $E = \bigcup_{i=1}^{n} E_i$, with $(\{v_1, \ldots v_i\}, E_i)$ the cycle of length $i$.
		Since $\majority{x} = 1$, then $\val{A(x)}{v_1} = 1$.
		Two cases are possible.
\begin{enumerate}
		\item $\val{x}{v_1} = 1$, therefore  $\val{A(x)}{v_2} = 1$.
		Now, since $\val{A(x)}{v_i} = \val{x}{v_{i-1}}$ for all $i$ between $2$ and $n$, we deduce that $\nbAx{T}{\{v_3, \ldots v_n\}} = \nbx{T}{\{v_2,\ldots, v_{n-1}\}} \ge \nbx{x}{\{v_2,\ldots, v_{n}\}} - 1$. 
		Consequently, since \\ $\nbAx{T}{\{v_2, \ldots v_n\}} = \nbAx{T}{\{v_3, \ldots v_n\}} + 1$, it follows that \\ $\nbAx{T}{\{v_2, \ldots v_n\}} \ge \nbx{x}{\{v_2,\ldots, v_{n}\}}$.
        
		At this point, three subcases are possible: $\val{x}{v_2} = 1$ and $\val{x}{v_n} = 0$;  $\val{x}{v_2} = 1$ and $\val{x}{v_n} = 1$; or $\val{x}{v_2} = 0$.
		For the first one, we have $\nbAx{T}{\{v_2, \ldots v_n\}} = \nbx{x}{\{v_2,\ldots, v_{n}\}} + 1$, and for the second, \\ $\nbAx{T}{\{v_2, \ldots v_n\}} = \nbx{x}{\{v_2,\ldots, v_{n}\}}$. From both, we conclude that $\majority{A(x)} = 1 = \val{A(x)}{v_1} = \val{A(x)}{v_2}$, and inductively, we go back to one of the previous subcases. 
		Thus, once again, since $\val{A(x)}{v_i} = \val{x}{v_{i-1}}$ for all $i$ between $2$ and $n$, we deduce that $\nbAit{T}{V}{j} = n$ for all $j \ge n-3$.
		For the third subcase, we have $\majority{A(x)} = 1 = \val{A(x)}{v_1} = \val{A(x)}{v_2}$, which leads back to the previous subcases, and inductively,  $\nbAit{T}{V}{j} = n$ for all $j \ge n-2$.
		
		\item We assume that $\val{x}{v_1} = 0$, therefore since $\val{A(x)}{v_i} = \val{x}{v_{i-1}}$ for all $i$ between $2$ and $n$, it follow that $\nbAx{T}{\{v_2, \ldots, v_n\}} = \nbx{T}{\{v_2, \ldots, v_{n-1}\}} \ge \nbx{T}{\{v_2, \ldots, v_n\}} - 1$. 
		In addition, $ \nbx{T}{\{v_2, \ldots, v_{n}\}} \ge \psup$. 
		Consequently, $\nbAx{T}{\{v_2, \ldots, v_n\}} \ge \psup - 1$. 
		Then, it follow that $\nbAx{T}{V} \ge \psup$. 
		By the first case, we conclude that $\nbAit{T}{V}{j} = n$ for all $j \ge n$.  
\end{enumerate}
		By applying the same reasoning now with $\majority{x} = 0$, the lemma follows.
	\end{proof}
	
	Let $n \ge 5$ be an odd number and $\mathcal{A}_n$ be the set of MBANs able to solve DCT, such that the graph of each $A \in \mathcal{A}_n$ has $n$ nodes and contains at least one node with in-degree $n$. 
	Remark that the MBAN of $\mathcal{A}_n$ with the minimal number of edges in its graph is the complete cycle MBAN of $n$ nodes. 
	
	We say that an MBAN is \emph{non-omniscient} if the in-degree of its graph $G=(V,E)$ is less than $|V|$; more formally, for each $v \in V, |N_G(v)| < |V|$.  
	Thus, if we want to reduce again the number of edges in the graph of the MBAN able to solve DCT, we have to study non-omniscient MBANs (actually, we do not yet know if there exists an MBAN family of bounded in-degree able to solve DCT, which would be \emph{local} in a stronger sense). Accordingly, we now introduce three classes of this kind of MBAN, all of them with a number of edges larger than $2 n$. 
	The first class (Sections~\ref{house_MBAN}) is one where the automata of the MBANs depend on $n - 2$ automata, whereas the second class is one where half of the automata in the MBAN depend on $n-2$ automata and the other half depend only on one automaton. 
	
	An interesting future work could be to show that the complete cycle MBAN of $n$ nodes is the network of $n$ automata able to solve DCT whose graph is the one with a minimal number of edges.  

	\subsection{Class of complementary-left-right MBAN}
\label{doubl_cor}

In this section, we talk about a network that has been found experimentally. 
The reason why this network solves DCT appears to us to be purely technical. 
Therefore, we cannot provide insightful intuition on how it works. 

Let $n \ge 7$ be an odd integer and $V = \{0, \ldots, n-1\}$. 
From $V$ we define a the following sequences: $U = (0,2,\ldots, \pinf - 1)$, $R = (1,\pinf, \ldots, n-1)$; $S = (0,1,\ldots, \pinf - 2)$ and $T = (\pinf - 1, \ldots, n-1)$, the latter two useful for defining the neighbourhood of $U$ and $R$, respectively. 
We denote by $S_i$ (resp., $T_i,U_i, R_i$) the $i$-th element of $S$ (resp., $T_i,U_i, R_i$)

Let $G = (V,E)$ be the graph such that in $E$ : 
\begin{itemize}[nosep]
    \item each node $U_i$ in $U$ have all nodes of $V \backslash \{S_{i + 1 \mod |S|}, S_{i + 2 \mod |S|} \}$ for predecessors,
    \item each node $R_i$ in $R$ have all nodes of $V \backslash \{T_{i -1 \mod |T|}, T_{i -2 \mod |T|}\}$ for predecessors
\end{itemize}
See Figure~\ref{fig:complemetary_corolla} for an example.
Since the complement of this graph is more readable than the graph itself, we also provide it. 
In addition, the proofs are easier to follow from the point of view of the complement.  

We say that the MBAN whose graph is $G$ is the \emph{complementary-left-right} MBAN of $n$ nodes, and that the class of complementary-left-right MBANs is the set of all the complementary-left-right MBANs.
The following algorithm can generate this kind of network for any size $n$: 

    \begin{enumerate}[nosep]
        \item Create $n$ nodes labeled from $0$ to $n-1$;
        \item Add all the possible arcs; 
        \item Create $4$ arrays $U= [0,2,\ldots, \pinf-1], R= [1, \pinf,\ldots, n -1], S = [0,\ldots, \pinf-2]$ and $T = [\pinf-1, \ldots, n-1]$;
        \item Remove the arcs $(S[i+1 \mod |S|], U[i])$ and $(S[i+2 \mod |S|], U[i])$ for each $i$ from $0$ to $|S|-1$;
        \item Remove the arcs $(T[i-1 \mod |T|], R[i])$ and $(T[i-2 \mod |T|], R[i])$ for each $i$ from $0$ to $|T|-1$.
    \end{enumerate}

\begin{figure}
	\centering
	\begin{minipage}[b]{0.5\linewidth}
   \centering
    \begin{tikzpicture}
	\tikzstyle{automate} = [fill=white, draw,circle,inner sep=0pt,minimum size=16pt]

	\tikzstyle{arc} = [-{>[length=1mm]}]
	
	\node[automate] (a2) at (-1,0.5)  {2};
	\node[automate] (a0) at (1,2)     {0};
	\node[automate] (a3) at (0,-1)    {3};
	\node[automate] (a1) at (0,1.5)   {1};
	\node[automate] (a6) at (2,1.5)   {6};
	\node[automate] (a5) at (2,-1)    {5};
	\node[automate] (a4) at (3,0)     {4};

    \draw[arc] (a0) to[loop above]  (a0);
    \draw[arc] (a0) to  (a1);
    \draw[arc] (a0) to[bend right]  (a4);
    \draw[arc] (a0) to  (a5);
    \draw[arc] (a0) to  (a6);

    \draw[arc] (a1) to[loop above]  (a1);
    \draw[arc] (a1) to  (a3);
    \draw[arc] (a1) to  (a4);
    \draw[arc] (a1) to  (a5);
    \draw[arc] (a1) to  (a6);

    \draw[arc] (a2) to  (a0);
    \draw[arc] (a2) to  (a1);
    \draw[arc] (a2) to  (a2);
    \draw[arc] (a2) to  (a5);
    \draw[arc] (a2) to  (a6);

    \draw[arc] (a3) to  (a0);
    \draw[arc] (a3) to  (a1);
    \draw[arc] (a3) to  (a2);
    \draw[arc] (a3) to[loop left]  (a3);
    \draw[arc] (a3) to  (a6);

    \draw[arc] (a4) to [bend left] (a0);
    \draw[arc] (a4) to  (a1);
    \draw[arc] (a4) to  (a2);
    \draw[arc] (a4) to  (a3);
    \draw[arc] (a4) to[loop right]  (a4);

    \draw[arc] (a5) to  (a0);
    \draw[arc] (a5) to  (a2);
    \draw[arc] (a5) to  (a3);
    \draw[arc] (a5) to  (a4);
    \draw[arc] (a5) to[loop right]  (a5);

    \draw[arc] (a6) to  (a0);
    \draw[arc] (a6) to  (a2);
    \draw[arc] (a6) to  (a4);
    \draw[arc] (a6) to  (a5);
    \draw[arc] (a6) to[loop above]  (a6);

\end{tikzpicture}
\end{minipage}\hfill
\begin{minipage}[b]{0.5\linewidth}
   \centering
    \begin{tikzpicture}
	\tikzstyle{automate} = [draw,circle,inner sep=0pt,minimum size=16pt]

	\tikzstyle{arc} = [-{>[length=1mm]}]
	
	\node[automate] (a2) at (-1,0.5)  {2};
	\node[automate] (a0) at (1,0.5)   {0};
	\node[automate] (a3) at (0,-1)    {3};
	\node[automate] (a1) at (3,0.5)   {1};
	\node[automate] (a6) at (4.5,0.5) {6};
	\node[automate] (a5) at (3,-1)    {5};
	\node[automate] (a4) at (4.5,-1)  {4};

    \draw[arc,loop below] (a0) to   (a0);
    \draw[arc] (a1) to  (a0);
    \draw[arc] (a0) to  (a2);
    \draw[arc] (a1) to [bend right] (a2);
    \draw[arc] (a5) to  (a1);
    \draw[arc] (a6) to  (a1);
    \draw[arc] (a2) to  (a3);
    \draw[arc] (a6) to  (a3);
    \draw[arc] (a2) to [bend right] (a4);
    \draw[arc] (a3) to [bend right] (a4);
    \draw[arc] (a3) to  (a5);
    \draw[arc] (a4) to  (a5);
    \draw[arc] (a4) to  (a6);
    \draw[arc] (a5) to  (a6);
\end{tikzpicture}
\end{minipage}
	\caption{The complementary-left-right MBAN with $7$ nodes on the left and its complement on the right.  According to the definition, $U = (0,2)$, $R = (1,3,4,5,6)$, $S = (0,1)$ and $T = (2,3,4,5,6)$. For example, the node $0$ has for in-neighbours all the nodes except $S_0 = 0$ and $S_1 = 1$; while the node $1$ has for in-neighbours all the nodes except $T_4 = 6$ and $T_3 = 5$.}\label{fig:complemetary_corolla}
\end{figure}
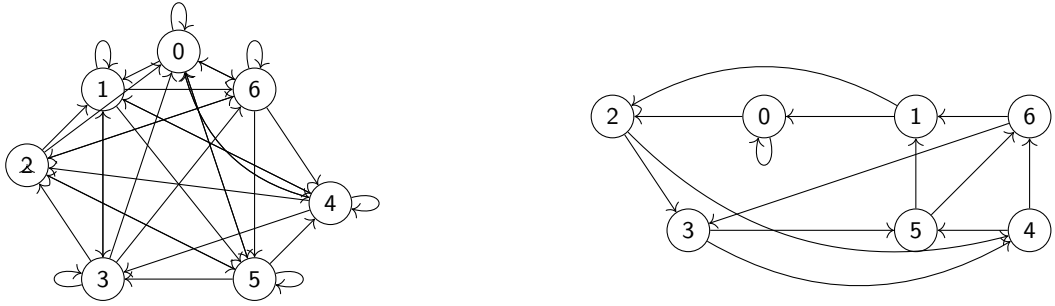

\ \\
Now we have two main goals. The first is to show that a complementary-left-right MBAN is able to solve DCT. The second is to provide a bound for the number of iterations needed by such MBANs to converge. 

So, let $A$ be a complementary-left-right MBAN, $G = (V,E)$ its graph, $n$ the number of automata in $A$, and $x$ one of its configurations such that $\majority{x} = 1$.
In order to prove that $A$ converges to $1$, we start by showing that $\nbAit{T}{V}{i} \le \nbAit{T}{V}{i+1}$ for any integer $i \ge 0$.  
To this end, we remark that two cases are possible: either,  the number of 1s in $x$ is at least $\psup + 1$ (\emph{i.e.}, $\nbx{T}{V} \ge \psup + 1$), or the number of 1s in $x$ is $\psup$ (\emph{i.e.}, $\nbx{T}{V} = \psup$). 
First, we assume that $\nbx{T}{V} \ge \psup + 1$.

\begin{lemma}\label{lemma:convergence_halft_plus_one}
    Let $A$ be an MBAN of $n$ automata whose graph $G = (V, E)$ is such that the in-degree of each node of $V$ is at least $n-2$. 
	If $\nbx{T}{V} \ge \psup + 1$, then $\nbAx{T}{V} = n$.
\end{lemma}

\begin{proof}
    Let us assume that $\nbx{T}{V} \ge \psup + 1$. 
    Since $\nbx{T}{V} \ge \psup + 1$, it follows that $\nbx{F}{V} < \pinf$. 
    Remark that the upper half of the in-degree of $v$ is at least $\pinf$ for each node $v$ of $V$. 
    Therefore, $\val{A(x)}{v} = 1$ for each $v$ in $V$.
\end{proof}

Now, assuming that $\nbx{T}{V} = \psup$, we distinguish three cases:
\begin{enumerate}
\item All the nodes of $S$ are 0 in $x$ (conversely, $\nbx{T}{T} = \psup$); \item All the nodes of $S$ are 1 in $x$ ($\nbx{T}{S} = \pinf - 1$); or 
\item Only some nodes of $S$ are 1 in $x$ ($0 < \nbx{T}{S} < \pinf - 1$). 
\end{enumerate}

--- For Case 1, if we assume that $\nbx{T}{T} = \psup$, then since each node $u$ of $U$ admits that all the nodes of $T$ are predecessor ($N_G(u) \cap T = T$), therefore $\nbAx{T}{U} = |U| = \pinf - 1$. 
Moreover, since $|T| = \psup + 1$, it follows that  $\nbx{F}{T} = 1$. 
Thus, $\nbAx{T}{R} = 2$.
Indeed, since each node of $T$ is the predecessor of exactly $\pinf$ nodes of $R$, two nodes of $R$ admit the $\psup$ nodes of $T$ whose value is 1. 
Hence, $\nbAx{T}{V} = \pinf - 1 + 2 = \psup$. 
In addition, $A(x)$ is either in the Cases 2 or 3.    

--- For Case 2, assuming that $\nbx{T}{S} = \pinf - 1$ then $\nbx{T}{T} = 2$. 
Let $t_1,t_2$ be the nodes of $T$ such that $\val{x}{t_1} = \val{x}{t_2} = 1$. From the definition of the graph, there is at most one node $a$ of $R$ whose the neighbourhood does not contain $t_1$ and $t_2$ ($N_G(a) \cap \{t_1,t_2\} = \emptyset$). 
Moreover, since the neighbourhood of each node $r$ of $R$ contains $S$ ($N_G(r) \cap S = S$), it follows that $\val{A(x)}{r} = 1$ for every node $r$ of $R$ that admits $t_1$ or $t_2$ in its neighbourhood ($N_G(r) \cap \{t_1,t_2\} \neq \emptyset$). 
Thus, $\nbAx{T}{R} \ge |R| - 1 = \psup$.
Also, since the neighbourhood of each node in $u$ is composed by $|S| - 2$ nodes of $S$ plus $|T|$ nodes of $T$, therefore $\nbx{T}{N_G(u)} = \pinf - 1$. Thus, $\nbAx{T}{U} = 0$. 
We conclude that, $A(x)$ is either in Cases 1 or 3.

%

%

--- For Case 3, let us assume that there exists an integer $0 < i < \pinf - 1$ such that $\nbx{T}{S} = i$, and proceed by introducing a technical lemma that will be used for deriving a lower bound for $\nbAx{T}{U}$ and $\nbAx{T}{R}$.
For this, we define that, given a sequence $\alpha = (\alpha_0, \ldots, \alpha_m)$, $p$ is a \emph{periodic consecutive subsequence} of $\alpha$ if it contains at least two elements and there exists $j \in \{0, \ldots, m\}$ such that $p = (\alpha_j, \alpha_{(j+1) \mod |p|}, \ldots, \alpha_{(j + |p| - 1) \mod |p|})$.

\begin{lemma}\label{lemma:peridoic_sub_seq}
	Let $\alpha$ be a sequence of $m$ elements.
	Let $|\alpha| - 1 > a > 1$ be an integer and $\beta$ be a subsequence of $a$ elements of $\alpha$.
	Let $m_1$ be the number of periodic consecutive subsequences of $\alpha$ with size $m-2$ that contains at least $a - 1$ elements of $\beta$, and $m$ be the number of periodic consecutive pairs $\alpha$ in $\beta$. Then $m_1 = m - m_3$.
\end{lemma}

\begin{proof}
	Let $P$ be the set of periodic consecutive subsequences of $\alpha$ with size $m - 2$ containing at most $a - 2$ elements of $\beta$, and $p$ be an element of $P$. 
	Since $\beta$ is a subsequence of $\alpha$, $p$ contains exactly $a - 2$ elements of $\beta$. 
	In the other case, the length of $p$ is less than $m-2$.
	Therefore $|P| = m - m_1$.
	
	Let $C$ be the set of periodic consecutive pairs of $\alpha$ in $\beta$. 
	Thus, the two absent elements in $p$ of $\alpha$ are in $\beta$.
	Therefore, since  $\alpha \backslash p$ is a periodic consecutive pair, for each element of $P$ there exists an element of $C$. 
	Reversely, if we consider an element $c$ of $C$, directly we have $\mathcal{S}  \backslash c$ is a periodic consecutive sequence of length $m-2$ with at most $a-2$ elements of $\beta$. 
	Thus, $m_3 = |P|$.
	The overall conclusion therefore is that $m_1 = m - m_3$.
\end{proof}

\begin{lemma}\label{lemma:double_corolla_between}
	Let $i$ be an integer between $1$ and $\pinf - 2$.
	If $\nbx{T}{S} = i$ then $\nbAx{T}{V} \ge \psup + 1$.
\end{lemma}

\begin{proof}
	We consider two possibilities, according to the value of $i$:
    
	--- $i = 1$:
	Hence, $\nbx{T}{T} = \pinf$ and,  
	consequently, since $N_G(u) \cap T = T$ for each $u \in U$, $\nbAx{T}{U} = |U| = \pinf - 1 = \pinf - i$. 
	
    --- $i > 1$:
	Let $u \in U$ be a node and notice that $N_G(u) \cap T = T$. Thus, $\nbx{T}{N_G(u) \cap T} = \psup - i$.
	We deduce that $\val{A(x)}{u} = 1$ if and only if $u \in U_o$ with $U_o := \{u_o \in U \mid \nbx{T}{N_G(u_o) \cap S} \ge i - 1\}$.
	Indeed, if $\nbx{T}{N_G(u) \cap S} \le i - 2$, then $\nbx{F}{N_G(u) \cap S} \ge \pinf - 1 - i$, and since   $\nbx{F}{N_G(u) \cap T} = i + 1$, it follows that $\nbx{F}{N_G(u)} \ge \pinf$.  
	Thus, if we consider $U_o$ as a subsequence of $U$ and since the maximum number of periodic consecutive pairs in $U_o$ is $i - 1$, from Lemma~\ref{lemma:peridoic_sub_seq}, we have that $|U_o| \ge \pinf - 1 - (i - 1) = \pinf - i$. 
	Therefore, the conclusion is that, $\nbAx{T}{U} \ge \pinf - i$.
	
	Now, let $r$ be a node of $R$.
	Note that $\val{A(x)}{r} = 1$ if and only if $r \in R_o$ with $R_o := \{r_o \in R \mid \nbx{T}{N_G(r_o) \cap T} \ge \pinf - i\}$.
	Indeed, if $\nbx{T}{N_G(r) \cap T} \le \pinf - i - 1$, then $\nbx{F}{N_G(r) \cap T} \ge i + 1$, and since  $\nbx{F}{N_G(r) \cap S} = \pinf - i - 1$, it follows that $\nbx{F}{N_G(u)} \ge \pinf$.  
	Thus, if we consider $R_o$ as a subsequence of $R$ and since the maximum number of periodic consecutive pairs in $R_o$ is $\psup - i - 1 = \pinf - i$, from Lemma~\ref{lemma:peridoic_sub_seq}, we have that $|R_o| \ge \psup + 1 - (\pinf - i) = i + 2$. 
	Hence, we conclude that $\nbAx{T}{V} \ge \pinf - i + i + 2 = \psup + 1$. 
\end{proof}

At this point, we have shown that $\nbAit{T}{V}{i} \le \nbAit{T}{V}{i+1}$, for every integer $i \ge 0$. 
Therefore, if $A$ converges then $A$ convergence to the correct configuration. 
It thus remains to prove that it converges. 

\begin{theorem}
    Each complementary-left-right MBAN is able to solve  DCT and it converges after at most $4$ iterations. 
\end{theorem}

\begin{proof}
    By Lemma~\ref{lemma:convergence_halft_plus_one} if $\nbx{T}{V} \ge \psup + 1$ then $A$ converges to the all 1s configuration. 
    Moreover, by Lemma~\ref{lemma:double_corolla_between}, $\nbx{T}{x} = i$ for an $i$ between $1$ and $\pinf - 2$ then $\nbAx{T}{V} \ge \psup + 1$ implying that $A$ converges in two iteration.
    Two cases remain.
    
    --- $\nbx{T}{V} = \psup$ and $\nbx{T}{S} = \pinf - 1$,  we have that $\nbAx{T}{R} = \psup$. 
    At this point, there are again two cases depending on the value of $R  \cap S$ in A(x) \ie the value of $\val{A(x)}{1}$.
    So, either $\val{A(x)}{1} = 1$ (\ie $0 < \nbx{T}{S} < \pinf - 1$) and from the previous reasoning $A$ converges to the all 1s configuration in $3$ iterations.
    Or $\val{A(x)}{1} = 0$, then $\nbAx{T}{T \cap R} = \psup$ and consequently $\val{A(x)}{T_0} = 0$ ($T_0 \in T$ but $T_0 \notin R$).
    Also, $\nbAit{T}{V}{2} = \psup$.
    From the definition of the graph, since the two only nodes $r_1,r_2$ of $R$ such that $N_G(r_1)  \cap T_0 = N_G(r_2) \cap T_0 = \emptyset$ are the two last elements of $R$, namely $R_{|R| - 1}$ and $R_{|R| - 2}$, we have that $\val{\iteration{A}{2}{x}}{R_{|R| - 1}} = \val{\iteration{A}{2}{x}}{R_{|R| - 2}} = 1$. 
    Thus $\nbAit{T}{T}{2} = 2$ and $\nbAit{T}{U}{2} = \pinf - 1$.
    Consequently, $\val{\iteration{A}{2}{x}}{R_{1}} = 0$, and we deduce that $0 < \nbAit{T}{S}{2} < \pinf - 1$.
    We conclude that $A$ converges to the all 1s configuration in $4$ iteration.

    --- $\nbx{T}{V} = \psup$ and $\nbx{T}{T} = \psup$, then $\nbAx{T}{U} = |U|$. Again, two cases are possible depending on $\val{A(x)}{1}$.
    If $\val{A(x)}{1} = 0$, then $\nbAx{T}{S} = \pinf - 2$ and therefore $A$ converges to the configuration with all 1s in $3$ iterations; but if $\val{A(x)}{1} = 1$, this entails that $\nbAx{T}{S} = |U|$. 
    Thus, $\nbAit{T}{V}{2} = \psup$ and $\nbAit{T}{U}{2} = 0$. 
    Furthermore, $\val{\iteration{A}{2}{x}}{1} = 1$; indeed, $\nbAx{T}{N_G(1) \cap (\{1\} \cup S)} = \pinf - 1$. it follows that $\nbAit{T}{S}{2} = 1$ (the only element $S$ which is not in $U$ is $1$).
    We conclude that $A$ converges to the configuration with all 1s in $4$ iterations.

    Therefore, $A$ converges within at most $4$ iterations. 
    Applying the same reasoning for $\majority{x} = 0$, the theorem follows.
\end{proof}

	\subsection{Class of complementary-circle-triangle MBAN}
\label{house_MBAN}
Again, in this section, we talk about a network that has been found experimentally, and it is hard to provide global intuitions on its dynamics.

Let $n \ge 7$ be an odd integer and $V = \{0, \ldots, n-1\}$. From $V$ we define two sequences, $U = (3, \ldots, n-1)$  and $R = \{0,1,2\}$.

Let $G = (V,E)$ be the graph such that in $E$ :
\begin{itemize}[nosep]
    \item each node $i$ of $U$ has all nodes of $V \backslash \{i,i-1\}$ for predecessors,
    \item the node $0$ has all nodes of $V \backslash\{2,n-1\}$ for predecessors,
    \item the nodes $1$ and $2$ have all nodes of $V \backslash \{0,1\}$ for predecessors.
\end{itemize}

See Figure~\ref{fig:house_network} for an example.
We refer to the MBAN with $G$ as the \emph{complementary-circle-triangle} MBAN of $n$ nodes, and to the corresponding class as the set of all complementary-circle-triangle MBANs.

The following algorithm can generate this kind of network for any size $n$: 

    \begin{enumerate}[nosep]
        \item Create $n$ nodes labeled from $0$ to $n-1$;
        \item Add all possible arcs; 
        \item Remove the arc $(i,i+1 \mod n)$ for each node $i$;
        \item Remove the arc $(i,i)$ for each node $i$ except for nodes $0$ and $2$;
        \item Remove the arcs $(0,2)$ and $(2,0)$.
    \end{enumerate}

\begin{figure*}
	\centering
	\begin{minipage}[b]{0.5\linewidth}
   \centering
   \begin{tikzpicture}
   \tikzstyle{automate} = [draw,circle,inner sep=0pt,minimum size=16pt]
	\tikzstyle{arc} = [-{>[length=1mm]}]
	
	\def\sommets{{3,2,0,6,5,4}} 
	
	\foreach \i in {0,...,5} {
		\pgfmathtruncatemacro{\nexti}{\sommets[\i]};
		\node[automate] (A\nexti) at ({60*\i}:2) {\pgfmathparse{\sommets[\i]}\pgfmathresult};
	}
	
	\node[automate] (A1) at (0,3) {1};
	
    \draw[arc] (A0) to[loop above] (A0);
    \draw[arc] (A0) to (A3);
    \draw[arc] (A0) to (A4);
    \draw[arc] (A0) to (A5);
    \draw[arc] (A0) to (A6);

    \draw[arc] (A1) to (A0);
    \draw[arc] (A1) to[bend left] (A3);
    \draw[arc] (A1) to (A4);
    \draw[arc] (A1) to (A5);
    \draw[arc] (A1) to[bend right] (A6);

	\draw[arc] (A2) to (A1);
    \draw[arc] (A2) to[loop above] (A2);
    \draw[arc] (A2) to (A4);
    \draw[arc] (A2) to (A5);
    \draw[arc] (A2) to (A6);

    \draw[arc] (A3) to (A0);
	\draw[arc] (A3) to[bend right] (A1);
    \draw[arc] (A3) to (A2);
    \draw[arc] (A3) to (A5);
    \draw[arc] (A3) to (A6);

    \draw[arc] (A4) to (A0);    
	\draw[arc] (A4) to (A1);
    \draw[arc] (A4) to (A2);
    \draw[arc] (A4) to (A3);
    \draw[arc] (A4) to (A6);

    \draw[arc] (A5) to (A0);
	\draw[arc] (A5) to (A1);
    \draw[arc] (A5) to (A2);
    \draw[arc] (A5) to (A3);
    \draw[arc] (A5) to (A4);

	\draw[arc] (A6) to[bend left] (A1);
    \draw[arc] (A6) to (A2);
    \draw[arc] (A6) to (A3);
    \draw[arc] (A6) to (A4);
    \draw[arc] (A6) to (A5);
\end{tikzpicture}
\end{minipage}\hfill
\begin{minipage}[b]{0.5\linewidth}
   \centering
\begin{tikzpicture}
	\tikzstyle{automate} = [draw,circle,inner sep=0pt,minimum size=16pt]
	\tikzstyle{arc} = [-{>[length=1mm]}]
	
	\def\sommets{{3,2,0,6,5,4}} 
	
	\foreach \i in {0,...,5} {
		\pgfmathtruncatemacro{\nexti}{\sommets[\i]};
		\node[automate] (A\nexti) at ({60*\i}:2) {\pgfmathparse{\sommets[\i]}\pgfmathresult};
	}
	
	\node[automate] (A1) at (0,3) {1};
	
	\foreach \i in {3,...,6} {
		\pgfmathtruncatemacro{\pred}{\i-1};
		\draw[arc] (A\pred) to[bend left, looseness=0.5] (A\i);
	}
	
	\draw[arc] (A3) to[loop right] (A3);
	\draw[arc] (A4) to[loop below] (A4);
	\draw[arc] (A5) to[loop below] (A5);
	\draw[arc] (A6) to[loop left] (A6);
	
	\draw[arc] (A6) to[bend left, looseness=0.5] (A0);
	\draw[arc] (A0) to[bend left] (A2);
	\draw[arc] (A0) to (A1);
	\draw[arc] (A1) to[loop above] (A1);
	\draw[arc] (A1) to (A2);
	\draw[arc] (A2) to[bend left] (A0);
\end{tikzpicture}
\end{minipage}
	\caption{On the left, the graph of the complementary-circle-triangle MBAN with $7$ nodes; and on the right, its complement.  
    According to the definition, $U = (3,4,5,6)$ and $R = (0,1,2)$. }\label{fig:house_network}
\end{figure*}
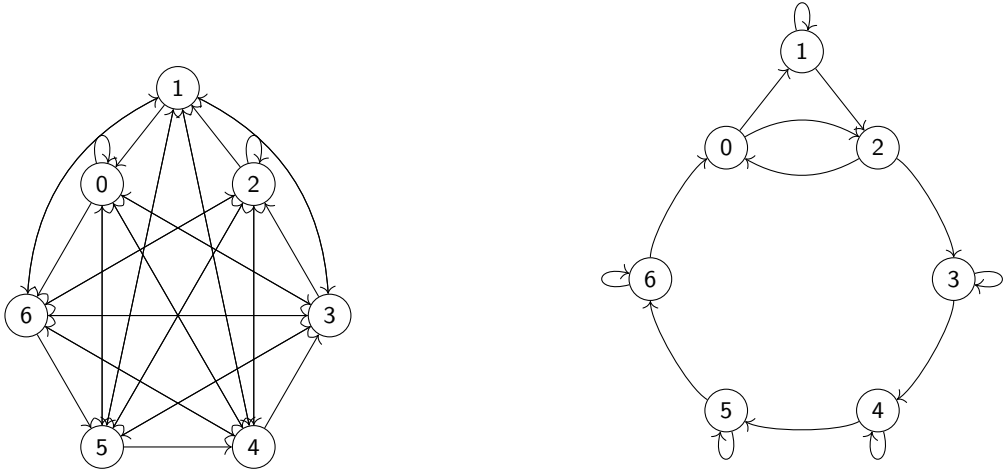

\ \\
Our goal in this section is to prove that any circle-triangle MBAN solves DCT. 
Let us then assume that $A$ is a circle-triangle MBAN and $x$ is a configuration of $A$ such that $\majority{x} = 1$. We will employ the same method as in the previous section.
We start by proving that $\nbAit{T}{V}{i} \le \nbAit{T}{V}{i+1}$ for any $i$. 
For this, we assume, without loss of generality, that $\nbx{T}{V} = \psup$. 
Indeed, if $\nbx{T}{V} \ge \psup + 1$ and, by Lemma~\ref{lemma:convergence_halft_plus_one}, $A$ converges to the fixed point of 1s. 

Let $S = (2,\ldots, n-1)$, and $\alpha$ be the subsequence of $S$ that contains all nodes $s \in S$ such that $\val{x}{s} = 1$. 
Remark that $\nbAx{F}{U}$ is the number of consecutive pairs of $S$ in $\alpha$. 
Indeed, since $\nbx{T}{V} = \psup$ and since $\left\lceil \cfrac{n-2}{2} \right\rceil = \pinf$, we need to remove two 1s so that the number of 0s becomes the majority. Thus, $\val{A(x)}{u} = 0$ if and only if $\val{x}{u-1} = \val{x}{u} = 1$, for all $u \in U$. 

\begin{lemma}\label{lemma:not_decresed_HN}
    If $\val{x}{0} = \val{x}{1}$, then $\nbx{T}{V} = \nbAx{T}{V}$. 
    Otherwise, $\nbAx{T}{V} \ge \psup + 1$. 
\end{lemma}

\begin{proof}
    We need to distinguish two cases, according to $\val{x}{0}$, which are addressed in the two following claims.

    \begin{claim}
    	Case 1: $\val{x}{0} = 0$. 
    	If $\val{x}{1} = \val{x}{0}$, then $\nbAx{T}{V} \ge \nbx{T}{V}$; otherwise, $\nbAx{T}{V} \ge \nbx{T}{V} + 1$. 
    \end{claim}

    \begin{claimproof}
        Let us proceed in two steps, initially assuming that $\alpha$ is not a consecutive subsequence of $S$. 
        Therefore, $\nbAx{F}{U} < \psup - \nbx{T}{R} - (1 - \val{x}{2})$.
        Note that $1 - \val{x}{2}$ is used to distinguish whether $\val{x}{2} = 1$ or $\val{x}{2} = 0$, and $\nbx{T}{R} = \val{x}{1} + \val{x}{2}$. 
        Then, $\psup - \nbx{T}{R} - (1 - \val{x}{2}) = \pinf - \val{x}{1}$. 
        Due to $\val{x}{0} = 0$, it follows that $\val{A(x)}{1} = \val{A(x)}{2} = 1$. 
        Consequently, $\nbAx{F}{V} < (1 - \val{A(x)}{0}) + \pinf - \val{x}{1}$.
        If $\val{A(x)}{0} = 1$ then, directly, $\nbAx{F}{V} < \pinf - \val{x}{1} \le \pinf$.
        And, if $\val{A(x)}{0} = 0$, then $\nbAx{F}{V} <  \psup - \val{x}{1} \le \psup$.
        Therefore $\nbAx{F}{V} < \psup$ or, alternatively, $\nbAx{T}{V} \ge \psup + \val{x}{1}$.
        
        As the second step, let us now consider that $\alpha$ is a consecutive subsequence of $S$ then, $\nbAx{F}{U} = \pinf - \val{x}{1}$. 
        Thus, since the smallest path between nodes $2$ and $n-1$ is larger than $\psup$, it follows that $2$ and $n-1$ cannot have the value 1 in $x$ at the same time. 
        Consequently, $\val{A(x)}{0} = 1$, which implies that $\nbAx{F}{V} = \pinf - \val{x}{1} \le \pinf$.
    \end{claimproof}

    \begin{claim}
	   Case 2: $\val{x}{0} = 1$.
	   If $\val{x}{1} = \val{x}{0}$ then $\nbAx{T}{V} \ge \psup$; otherwise, $\nbAx{T}{V} \ge \psup + 1$. 
    \end{claim}

    \begin{claimproof}
        Once again, let us proceed in two steps, first assuming that $\alpha$ is not a consecutive subsequence of $S$. 
        Then, $\nbAx{F}{U} < \psup - \nbx{T}{R} - (1 - \val{x}{2})$.
        Since $\nbx{T}{R} = \val{x}{0} + \val{x}{1} + \val{x}{2} = 1 + \val{x}{1} + \val{x}{2}$, it follows that  $\psup - \nbx{T}{R} - (1 - \val{x}{2}) = \pinf - \val{x}{1} - 1$. 
        Because $\val{x}{0} = 1$, this entails that $\val{A(x)}{1} = \val{A(x)}{2} = (1 - \val{x}{1})$. 
        Then, $\nbAx{F}{V} < (1 - \val{A(x)}{0}) + 2 \val{x}{1} + \pinf - \val{x}{1} - 1 = \pinf - \val{A(x)}{0} + \val{x}{1}$. 
        So, either $\val{A(x)}{0} = 1$ and $\nbAx{F}{V} < \pinf$, or $\val{A(x)}{0} = 0$, which implies that $\val{x}{1} = 1$ and $\nbAx{F}{V} \le \pinf$. 
        
        Now, considering that $\alpha$ is a consecutive subsequence of $S$, then $\nbAx{F}{U} = \pinf - \val{A(x)}{0} + \val{x}{1}$ and nodes $n-1$ and $2$ cannot take on vale 1 in $A(x)$. 
        Hence, $\val{A(x)}{0} = 1$, and it follows that either $\val{x}{1} = 0$ and  $\nbAx{F}{V} < \pinf$, or $\val{x}{1} = 1$ and  $\nbAx{F}{V} = \pinf$.
    \end{claimproof}
\end{proof}

It remains to show that $A$ converges.

\begin{theorem}
    Each complementary-circle-triangle MBAN is able to solve DCT,
    and it converges in $5$ iterations.
\end{theorem}

\begin{proof}
    As already explained, if $\nbx{T}{V} \ge \psup + 1$ then, by Lemma~\ref{lemma:convergence_halft_plus_one}, $A$ converges in 1 iteration. 
    Therefore, we assume that $\nbx{T}{V} = \psup$. 
    Moreover, by Lemma~\ref{lemma:not_decresed_HN}, if $\val{x}{0} \neq \val{x}{1}$ then $\nbAx{T}{V} \ge \psup + 1$ and $A$ converges in $2$ iterations.
    Thus, we assume also that $\val{x}{0} = \val{x}{1}$. 
    We distinguish again two cases, according to the value of $\val{x}{0}$. 

	First, let us consider $\val{x}{0} = 0$.
    Since if $\nbAx{T}{V} \ge \psup +1$, then $A$ converges in $2$, and we assume also that $\nbAx{T}{V} = \psup$.
    Three sub-cases are possible:
    
    --- $\val{x}{2} = 0$: Consequently $\val{A(x)}{0} = 1$ and therefore $\nbAx{T}{R} = 3$.
    Thus, $\nbAx{F}{V} = \nbAx{F}{U} < \psup - 1 = \pinf$, and then $\nbAx{T}{V} > \psup$ and $A$ converges in $2$ iterations.
    
    --- $\val{x}{2} = 1 = \val{x}{n-1} = 1$: Then, $\val{A(x)}{0} = 0$ and $\val{A(x)}{1} = 1$, and it follows that $A$ converges in $3$ iterations.
    
    --- $\val{x}{2} = 1$ and $\val{x}{n-1} = 0$:
    Therefore, $\val{A(x)}{n-1} = 1$ and $\nbAx{T}{R} = 3$.
    Thus, either $\nbAit{T}{V}{2} \ge \psup +1$ and $A$ converges in $3$ iterations, or $\nbAit{T}{R}{2} = 0$ and $A$ converges in at most $5$ iterations.

    Now, let us assume that $\val{x}{0} = 1$. 
    Since if $\nbAx{T}{V} \ge \psup +1$, then $A$ converges in $2$ iterations, we assume also that $\nbAx{T}{V} = \psup$, and two further 
possibilities arise:

    --- $\val{x}{2} = 0$ or $\val{x}{n-1} = 0$: Then, $\val{A(x)}{0} = 1$ and $\val{A(x)}{1} = 0$ and, consequently, $A$ converges in at most $3$ iterations.
    
    --- $\val{x}{2} = 1 = \val{x}{n-1}$: Thus, $\nbAx{T}{R} = 0$ and $A$ converges in at most $3$ iterations. 
    
    By applying the analogous reasoning above to $\majority{x} = 0$ we conclude the proof.
\end{proof}

	\subsection{Class of MBANs with two intersecting cycles}

In the previous subsections, we showed two classes of non-omniscient MBANs able to solve DCT, and whose convergence is fast. 
However, their number of edges is large: $n^2 - 2n$. 
Here, we focus our attention on a class of non-omniscient MBANs able to solve DCT, with a smaller number of edges, but with a slower convergence speed (linear in $n$). 
The idea of these networks is to modify the complete circle network by transforming the root of the path into a set of nodes with an in-degree equal to $n-2$.  
More precisely, we define two sets of nodes, namely $U$ and $R$, such that $U \cap R \neq \emptyset$. 
The nodes of $R$ play the role of the root path of the complete cycle graph. 
In other words, after some iteration, all the nodes of $R$ have as value the majority of the original configuration. 
Then, thanks to the intersecting point, this value is propagated along to the nodes of $U$. 

Let $n \ge 7$ be an odd integer and $V = \{0, \ldots, n-1\}$. 
Let $U = \{1,\ldots, \pinf\}$, $R = \{\psup, \ldots, n-2\}$, and $c$ an element of $R$ (the \emph{cross point}). 
We consider $G = (V,E)$ the graph such that in $E$  
\begin{itemize}[nosep]
	\item $(i-1, i) \in E$ for all $i \in U$,
	\item $(i - \pinf, j) \in E$ for all $i,j \in R$ with $i \neq j$,
	\item $(n-1, i) \in E$ for all $i \in R$,
	\item $(i, n-1) \in E$ for all $i \in V \backslash \{0,c\}$, and
	\item $(c,0) \in E$.
\end{itemize}

\ \\ 
Intuitively if $c \neq n-2$, the sequences $S = (1,\ldots,\pinf,c,0)$ and $T = (\psup,\ldots,n-1)$ are two cycles, and any $T_i$ $\in$ $T$, different to $n-1$, has all nodes of $G$ as predecessors, except $S_i$ and $0$, and node $n-1$ has for predecessors all nodes of $G$ except $c$ and $0$. See Figure~\ref{fig:cross_network_7_nodes} for an example.
The following algorithm allows the construction of the kind of network for any size $n$ and cross point $c$: 

    \begin{enumerate}[nosep]
        \item Create $n$ nodes labeled from $0$ to $n-1$;
        \item Add the arc $(i,i+1)$ for each $i$ from $0$ to $\pinf-1$;
        \item Add the arc $(i,j)$ for all nodes $i$ between $1$ and $n-1$ and all nodes $j$ between $\psup$ and $n-1$;
        \item Remove the arc $(i - \pinf,i)$ for each $i$ between $\psup$ and $n-2$;
        \item Remove the arc $(c,n-1)$;
        \item Add the arc $(c,0)$.
    \end{enumerate}

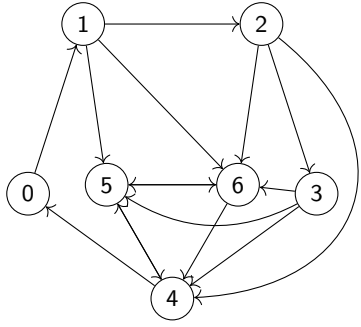
\begin{figure}
	\centering
	\begin{tikzpicture}
	\tikzstyle{automate} = [draw,circle,inner sep=0pt,minimum size=16pt]
	\tikzstyle{arc} = [-{>[length=1mm]}]
	
	\def\sommets{{3,2,1,0,4}} 
	
	\foreach \i in {0,...,4} {
		\node[automate] (A\i) at ({72*\i-18}:2) {\pgfmathparse{\sommets[\i]}\pgfmathresult};
	}
	
        \begin{scope}[shift={(0,-1)}]
                \node[automate] (A5) at (150:1) {5}; 
	        \node[automate] (A6) at (30:1) {6};  
        \end{scope}
	
	\foreach \i in {0,...,4} {
		\pgfmathtruncatemacro{\next}{mod(\i+1,5)}
		\draw[arc] (A\next) to (A\i);
	}
	
        \draw[arc] (A1) .. controls (0:3.2) and (-36:3.2) .. (A4);
	\draw[arc] (A2) to (A5);
	\draw[arc] (A2) to (A6);
	\draw[arc] (A1) to (A6);
	\draw[arc] (A0) to (A6);
	\draw[arc] (A0) to[bend left] (A5);
	\draw[arc] (A4) to (A5);
	\draw[arc] (A5) to (A4);
	\draw[arc] (A5) to (A6);
	\draw[arc] (A6) to (A5);
	\draw[arc] (A6) to (A4);
\end{tikzpicture}
	\caption{A graph of an MBAN with two intersecting cycles with 7 nodes and node $c = 4$.
    According to the definition, the cycles $S = (1,2,3,4,0)$ and $T=(4,5,6)$ have node $4$ as intersecting point. Remark that all nodes of $S$, except node $4$, have $1$ predecessor, and all nodes of $T$ have $n-2$ predecessors. 
    For examples, $1$ has for in-neighbor $1-1 = 0$, the node $5$ has for in-neighbors all the nodes except $0$ and $5-3 = 2$ and the nodes $6$ has for in-neighbors all the the nodes of the graph except $0$ and $4$.  
    }
    \label{fig:cross_network_7_nodes}
\end{figure}

\ \\
Let $A$ be an MBAN with such a graph. 
Then, we say $A$ is an MBAN with \emph{two intersecting cycles} of $n$ nodes, and define the corresponding class of all MBANs with two intersecting cycles.

Now, if $A$ is an MBAN with two intersecting cycles and $x$ is a configuration of $A$ such that $\majority{x} = 1$, for proving that $A$ is able to solve DCT we begin by showing that $\nbx{T}{V} \le \nbAx{T}{V}$ and, once again, two cases are possible: 
$\nbx{T}{V} \ge \psup + 1$, or $\nbx{T}{V} = \psup$. 

\begin{lemma}\label{lemma:cross_half_1_conv}
	If $\nbx{T}{V} \ge \psup + 1$, then $\nbAx{T}{V} \ge \nbx{T}{V}$ and $\nbAit{T}{V}{i} = n$ for all $i \ge \psup$. 
\end{lemma}

\begin{proof}
	We assume that $\nbx{T}{V} \ge \psup + 1$. 
	Then, $\nbAx{T}{R \cup \{n-1\}} = |R| + 1 = \pinf$. 
	Now, two cases are possible.
    
	--- $\nbx{T}{R \cup \{n-1\}} = \pinf$:\\ 
	Therefore, $\nbAx{T}{U} =  \nbx{T}{\{0\} \cup U} \ge \nbx{T}{\{0, \ldots, \pinf\}} - 1$. 
	However, since $\val{A(x)}{0} = 1$, it follows that $\nbAx{T}{\{0\} \cup U} = \nbAx{T}{U} + 1 \ge \nbx{T}{\{0\} \cup U}$; hence, $\nbAx{T}{V} \ge \nbx{T}{V}$. 
	Moreover, since $\nbAx{T}{R \cup \{n-1\}} = \pinf$, we remain in this case.
	Consequently, since $\val{x}{i} = \val{x}{i-1}$ for all $i \in U$, it turns out that $\nbAit{T}{\{0\} \cup U}{j} = \psup$ and $\nbAit{T}{V}{j} = n$ for all $j \ge \pinf$.
    
	--- $\nbx{T}{R \cup \{n-1\}} < \pinf$: \\
	Then, $\nbAx{T}{R \cup \{n-1\}} \ge \nbx{T}{R \cup \{n-1\}} + 1$. 
	Now, since $\nbAx{T}{\{0\} \cup U} \ge \nbx{T}{\{0\} \cup U} - 1$, it follows that $\nbAx{T}{V} \ge \nbx{T}{R \cup \{n-1\}} + 1 + \nbx{T}{\{0\} \cup U} - 1 \ge \nbx{T}{V}$.
	And, from the first case, $\nbAit{T}{V}{j} = n$ for all $j \ge \pinf + 1$
\end{proof}

For the latter case, we assume that $\nbx{T}{x} = \psup$, and remark that the nodes $0$ and $c$ have specific features. So much so that, in the next two lemmas, we need to consider two possibilities, according to $\val{x}{0}$, initially with value 0 and then with value 1.

\begin{lemma}\label{lemma:TCL_0} 
    Let $S = \{\pinf, \ldots, n-1\} \backslash \{c\}$.
    Let us assume that $\val{x}{0} = 0$ and $\nbx{T}{V} = \psup$.
	If $\nbx{T}{S} = |S| = \pinf$, then $\nbAx{T}{V} = \psup$; 
	otherwise, $\nbAx{T}{V} \ge \psup + 1$.
    In addition, in all cases, $\val{A(x)}{c} = 1$.
\end{lemma}

\begin{proof}
	Since $\val{x}{0} = 0$, we have $\nbx{F}{N_G(r)} \le \pinf - 1$ for all $r \in R u\cup \{n-1\}$. 
	Thus, $\nbAx{T}{R \cup \{n-1\}} = \pinf \ge \nbx{T}{R \cup \{n-1\}}$. 
	Therefore, $\nbAx{T}{\{0\} \cup R \cup \{n-1\}} = \pinf + \val{x}{c}$.
	Indeed, $\val{x}{c} = 1$ if and only if $\val{A(x)}{0} = 1$.
	
	Moreover, $\nbAx{T}{U} = \nbx{T}{U} - \val{x}{\pinf}$.
	However, since $\nbx{T}{U} = \nbx{F}{S'} + (1 - \val{x}{c}) + 1$, with $S' = (R \cup \{n-1\}) \backslash \{c\}$, it turns out that 
	$\nbAx{T}{V} = \pinf + \val{x}{c} + \nbx{F}{S'} + (1 - \val{x}{c}) + 1 - \val{x}{\pinf} = \psup + \nbx{F}{S'} + (1 - \val{x}{\pinf})$. 
	Hence, $\nbAx{T}{V} \ge \psup$ and $\nbAx{T}{V} = \psup$ if and only if $\val{x}{\pinf} = 1$ and $\nbx{F}{S'} = 0$.
	The lemma follows.
\end{proof}

Now, we assume that $\val{x}{0} = 1$, and for deriving the proof we proceed in two steps. First, we focus our attention on the nodes in $R$ then we conclude.

\begin{lemma}\label{lemma:TCL_1_R}
	Let us assume that $\val{x}{0} = 1$. Then, $$\nbAx{T}{R}~=~\nbx{T}{R}~+~\val{x}{n-1}~+~\val{x}{\pinf}~-~1.$$
\end{lemma}

\begin{proof}
	Let $R' = R \cup \{n-1\}$. 
	Since $|R'| = |U| = \pinf$ and $\nbx{T}{V \backslash \{0\}} = \pinf$, we deduce that there exists a bijection $f: R' \to U$, such that $\val{x}{i} = 1$ if and only if $\val{x}{f(i)} = 0$.
	 
	Let $I$ be the subset of $R'$ such that $i \in I$, if and only if $\val{x}{i - \pinf} = 0$ and $\val{x}{i} = 1$. 
	Without loss of generality, we can suppose that $f(i) = i - \pinf$ for all $i \in I$.
	Indeed, if there exists $i \in I$ such that $f(i)$ is not $i - \pinf$, we can just consider $g: R' \to U$ such that:
	\[ j \mapsto \begin{cases}
		i - \pinf &\textnormal{if } j = i,\\
		f(i), &\textnormal{if } j = f^{-1}(i - \pinf), \\
		f(j), &\textnormal{otherwise.} \\
	\end{cases}\]
	
	Let $Z$ be the subset of $R'$ such that $i \in Z$ if and only if $\val{x}{i} = 1$ and $\val{x}{i - \pinf} = 0$, and $Z'$ be the set $\{z' \mid \exists z \in Z, z' = \psup + z - 1\}$. 
	As in the previous point, we can suppose, without loss of generality that $f(i) = i - \pinf$. 
	
	We consider the subset $I'$ of elements of $R' \backslash I$ such that $i \in I'$ if and only if $\val{x}{i} = 1$ and $\val{x}{i - \pinf} = 1$, and also the subset $Z'$ of elements of $R' \backslash Z$ such that $i \in Z'$ if and only if $\val{x}{i} = 0$ and $\val{x}{i - \pinf} = 0$.
	It follows that $i \in Z'$ if and only if $j \in I'$ with $j = f^{-1}(i - \pinf)$; indeed, if $i \in Z'$ then $\val{x}{j} = 1$ (because $\val{x}{i - \pinf} = 0$). 
	And since $\val{x}{i} = 0$, we deduce that $i \neq j$. 
	Therefore, from the hypothesis on $f$, $j \notin I$, thus implying that $j \in I'$. 
	And if $j \in I'$, then $\val{x}{i - \pinf} = 0$ and $i \neq j$. 
	Thus, from the hypothesis on $f$, $i \notin I$ implies that $\val{x}{i} \neq 1$. Consequently, $\val{x}{i} = 0$ and $i \in Z'$.
	
	At this point, $\nbAx{T}{I \cap R} = |I \cap R| = \nbx{T}{I \cap R}$ and $\nbAx{F}{Z \cap R} = |Z \cap R| = \nbx{F}{Z \cap R}$. 
	Thus, if $\val{x}{n-1} \neq \val{x}{\pinf}$, then $n-1 \in I$ or $n-1 \in Z$. 
	This implies that $I', Z' \subseteq R$. 
	Consequently, $i - \pinf \in U \backslash \{\pinf\}$ for all $i \in I' \cup Z'$. 
	Therefore, for each $i \in I'$, the node $j = f^{-1}(i - \pinf)$ is in $Z'$ and $\nbx{F}{N_G(j)} \le \pinf - 1$; 
	also, for each $i \in Z'$, the node $j = f^{-1}(i - \pinf)$ is in $I'$ and $\nbx{T}{N_G(j)} \le \pinf - 1$.
	All this leads to $\nbAx{T}{Z'} = |I'| = |Z'| = \nbAx{F}{I'}$. 
	Finally, $\nbAx{T}{R} = \nbAx{T}{I \cap R} + \nbAx{T}{Z'} = \nbx{T}{I \cap R} + \nbx{T}{I'} = \nbx{T}{R}$.    
	
	Now, we assume $\val{x}{n-1} = \val{x}{\pinf}$. 
	Once again, without loss of generality, we can suppose that $f^{-1}(\pinf) - \pinf = f(n-1)$; 
	indeed, since $\val{x}{n-1} = \val{x}{\pinf}$, it turns out that $n-1 \in I'$ or $n-1 \in Z'$. 
	Therefore, $\val{x}{j} = (1- \val{x}{n-1})$, and  
	we deduce that $\val{x}{f(j)} = \val{x}{n-1}$, with $j = f^{-1}(\pinf)$. 
	Furthermore, $i \in I'$ or $i \in Z'$ and $\val{x}{i} = (1 - \val{x}{n-1})$ with $i = f(n-1) + \pinf$.
	Hence, $\val{x}{f(i)} = \val{x}{n-1} = \val{x}{f(j)}$, and we can just consider $g: R' \to U$ such that:
	\[ k \mapsto \begin{cases}
		\pinf &\textnormal{if } k = i,\\
		f(i) &\textnormal{if } k = j, \\
		f(k) &\textnormal{otherwise.} \\
	\end{cases}\]
	
	Analogously to the previous cases, $\nbAx{T}{Z' \backslash \{n-1, f^{-1}(\pinf)\}} = |Z'| - 1 = |I'| - 1 = \nbAx{F}{I' \backslash \{n-1, f^{-1}(\pinf)\}}$, and we conclude that $\nbAx{T}{R} = \nbx{T}{I} + |I'| - 1 + \val{A(x)}{f^{-1}(\pinf)}$.
	So, either $\val{x}{n-1} = 1$, which implies that $\val{A(x)}{f^{-1}(\pinf)} = 1$ and $|I'| = \nbx{T}{I' \cap R} + 1$ and, therefore, $\nbAx{T}{R} = \nbx{T}{I} + \nbx{T}{I'\cap R} + 1 = \nbx{T}{R} + 1$, 
	or $\val{x}{n-1} = 0$, thus implying that $\val{A(x)}{f^{-1}(\pinf)} = 0$ and $|I'| = \nbx{T}{I' \cap R}$.
	Consequently, $\nbAx{T}{R} = \nbx{T}{I} + \nbx{T}{I'\cap R} - 1 = \nbx{T}{R} - 1$.
\end{proof}

\begin{lemma}\label{lemma:TCL_1}
	Assuming that $\val{x}{0} = 1$ and \ $\nbx{T}{V} = \psup$, it is true that $\nbAx{T}{V} = \nbx{T}{V} + \val{x}{c}$.
\end{lemma}

\begin{proof}
	Since $\val{x}{0} = 1$, we deduce that $\val{A(x)}{n-1} = (1 - \val{x}{c})$. 
	In addition, since $\nbAx{T}{U} = \nbx{T}{U} - \val{x}{\pinf} + \val{x}{0}$, we have that $\nbAx{T}{U} =  \nbx{T}{U} - \val{x}{\pinf} + 1 = \nbx{T}{\{0\} \cup U} - \val{x}{\pinf}$, and also that $\val{A(x)}{0} = \val{x}{c}$.
	Additionally, by Lemma~\ref{lemma:TCL_1_R}, we have $\nbAx{T}{R} = \nbx{T}{R} + \val{x}{n-1} + \val{x}{\pinf} - 1 = \nbx{T}{R \cup \{n-1\}} +  \val{x}{\pinf} - 1$, which then allows us to conclude that: 
	\begin{align*}
		\nbAx{T}{V} &= \nbAx{T}{R} + \val{A(x)}{n-1} + \val{A(x)}{0} + \nbAx{T}{U}\\ 
		&= \nbx{T}{R \cup \{n-1\}} +  \val{x}{\pinf} - 1 + 1 - \val{x}{c} + \val{x}{c} + \nbx{T}{\{0\} \cup U} - \val{x}{\pinf}\\
		&= \nbx{T}{R \cup \{n-1\}} + \nbx{T}{\{0\} \cup U}\\
		&= \nbx{T}{V}.
	\end{align*}
\end{proof}

All that remains is to prove that $A$ converges, which is obtained in the following lemma. 

\begin{theorem}
    Each MBAN with two intersecting cycles is able to solve DCT,
    and it converges in $n+5$ iterations.
\end{theorem}

\begin{proof}
	Without loss of generality, let us assume that $\nbx{T}{V} = \psup$. Indeed, in the other case $\nbAit{T}{V}{i} = n$ for all $i \ge \psup$ by Lemma~\ref{lemma:cross_half_1_conv}.
	This leads to two possibilities, according to the value of $\val{x}{0}$, as we proceed below. 	

	--- $\val{x}{0} = 1$: 
	Let $p$ the be greater integer in $\{0, c - \pinf\}$  such that $\val{x}{p} = 1$ and $d = c - \pinf - p$; therefore, $\val{A^{i}(x)}{c} = 1 = \val{A^{i(x)}}{0}$ for all integer $0 \le i \le d$, and, from Lemma~\ref{lemma:TCL_1}, $\nbAit{T}{V}{i} = \psup$. 
	Thus, $\val{A^{d+1}(x)}{c} = 0$ and $\val{A^{d+1}(x)}{0} = \val{A^{d+1}(x)}{1} = \val{A^{d+1}(x)}{c+1 - \pinf} = 1$. 
	From Lemma~\ref{lemma:TCL_1}, we also have that $\nbAit{T}{V}{d+1} = \psup$, and remark that $c+1 - \pinf$ can be equal to $2$ or $\pinf$, because $n \ge 7$ and $\pinf = 3$.
	Therefore, at least one $j \in R \backslash \{c\}$ is such that $\val{A^{d+2}(x)}{j} = 0$, namely, $j = 1$ or $j = c+1$. 
	In addition, since  $\val{A^{d+1}(x)}{c} = 0$, we have $\val{A^{d+2}(x)}{0} = 0$. 
	By Lemma~\ref{lemma:TCL_0}, we deduce that $\nbAit{T}{V}{d+3} \ge \psup + 1$. 
	Thus, by Lemma~\ref{lemma:cross_half_1_conv}, $\nbAit{T}{V}{j} = n$ for all $j \ge d + 3 + \psup$. 
	And since $c - \pinf < \pinf$, $d$ we can be bounded by $\pinf$, thus entailing that $j \ge \pinf + \psup + 3 = n + 3$.

	--- $\val{x}{0} = 0$:
    Let $S = (\{\pinf\} \cup R \cup \{n-1\} )\backslash \{c\}$. 
    Remark that if $\nbx{T}{S} < \pinf$, then $A$ converges to 1s in $\psup + 1$ iterations. 
    Indeed, by Lemma~\ref{lemma:TCL_0}, $\nbAx{T}{V} \ge \psup + 1$ and therefore, by Lemma~\ref{lemma:cross_half_1_conv} $\nbAit{T}{V}{j} = n$ for all $j \ge \psup + 1$.
    Consequently, we assume that $\nbx{T}{S} = \pinf$. 
    Thus, from Lemma~\ref{lemma:TCL_0}, we deduce that $\nbAx{T}{V} \ge \psup$ and $\val{A(x)}{c} = 1$.
	Therefore, $\val{A^{2}(x)}{0} = 1$ and $\nbAit{T}{V}{2} \ge \psup$ by Lemma~\ref{lemma:TCL_0} if $\val{A(x)}{0} = 0$ or by Lemma~\ref{lemma:TCL_1} if $\val{A(x)}{0} = 1$.
	We deduce that, by the first case, $\nbAit{T}{V}{j} = n$ for all $j \ge n + 3 + 2 = n + 5$. 
\end{proof}

	\section{Concluding remarks}

All the effort that has been put in trying to solve the density classification problem throughout the years has yielded invaluable insights into the mechanisms by which global properties can emerge from purely local rules. The DCT, therefore, stands as a testament to the power and limitations of emergent computation, offering a rich landscape for future exploration in the ever-evolving field of complex systems.

Although it is known for various years that no solution for the problem can exist for automata networks employing a circular graph of connection pattern -- i.e., the one associated to standard cellular automata -- here we showed that, beyond some trivial possibilities derived from arbitrary connection patterns, non-trivial cases do exist, where only partial visibility of the nodes by the others is allowed, by presenting four families of them. Here we presented them and proved their functioning; in another text, many complexity issues associated to them will be given, since presenting them here would make this text too lengthy and hence cumbersome.

Beyond the families we identifed, others might exist. In fact, exhaustive computational experiments enumerating all possible networks, up to isomorphism, led to 10 solutions for networks with 3 nodes, and 7514 solutions for networks with 5 nodes. So, the number of possible solutions grows very fast, even though, as should be noticed, not every solution of a given size is necessarily generalisable to larger sizes, which significantly restricts the possible solutions.

The solution families we were able to identify with the local majority rule has a certain elegance, as this represents the global problem being solved by its local variant. But this lets us wondering about how rich would the space of possible solutions be if we would allow arbitrary rules to be employed, a testament of the flexibility automata networks give us beyond cellular automata.

\section*{Acknowledgements}

The authors thank Gilberto de Melo Junior for probing the algorithms given in the paper.
This work has been partially funded by the HORIZON-MSCA-2022-SE-01 project 101131549: ``Application-driven Challenges for Automata Networks and Complex Systems (ACANCOS)''. P.P.B.~also thanks the Brazilian agency CNPq (Conselho Nacional de Desenvolvimento Científico e Tecnológico) for the research grant PQ 303356/2022-7. K.P.~and M.R.~thank ANR-24-CE48-7504
 ALARICE project.

	\bibliographystyle{cas-model2-names}
	\bibliography{biblio}
	
\end{document}